\documentclass[12pt]{amsart}
\usepackage{geometry} 
\usepackage{graphicx}
\usepackage{amsfonts}
\usepackage{amsthm}
\usepackage{amsmath}
\usepackage{amssymb}
\usepackage[arrow,matrix,curve]{xy}
\usepackage{enumerate}
\usepackage{color}
\usepackage[normalem]{ulem}
\usepackage{tikz}
\usepackage{xspace}
\definecolor{light-blue}{rgb}{0.8,0.85,1}
\definecolor{light-red}{rgb}{1,.4,.4}
\definecolor{purp}{rgb}{.7,.3,1}
\definecolor{yel}{rgb}{1,1,.5}
\definecolor{cy}{rgb}{0,1,1}
\usepackage{colortbl}
\usepackage{multirow}
\usepackage{array}
\newtheorem{theorem}{Theorem}

\newtheorem{lemma}{Lemma}
\newtheorem{proposition}{Proposition}
\theoremstyle{definition}

\newtheorem{remark}{Remark}
\newcommand{\co}{\colon\,}
\newcommand{\bT}{\mathbb T}
\newcommand{\bR}{\mathbb R}
\newcommand{\bC}{\mathbb C}
\newcommand{\bZ}{\mathbb Z}

\newcommand{\bH}{\mathbb H}

\newcommand{\cH}{\mathcal H}
\newcommand{\cK}{\mathcal K}

\newcommand{\tG}{\widetilde G}
\newcommand{\ta}{\widetilde a}
\newcommand{\tb}{\widetilde b}
\newcommand{\tI}{\widetilde{\mathrm I}}
\newcommand{\tIA}{\widetilde{\mathrm{IA}}}

\newcommand{\lp}{\textup{(}}
\newcommand{\rp}{\textup{)}}

\newcommand{\Aut}{\operatorname{Aut}}

\renewcommand\Re{\operatorname{Re}}

\newcommand{\id}{\text{id}}
\title[Real Baum-Connes and T-duality]{Real Baum-Connes assembly and\\
T-duality for torus orientifolds}

\author{Jonathan Rosenberg}
\address{Department of Mathematics\\
University of Maryland\\
College Park, MD 20742-4015, USA} 
\email[Jonathan Rosenberg]{jmr@math.umd.edu}
\thanks{Partially supported by NSF grant DMS-1206159.}
  
\begin{document}

\begin{abstract}
We show that the real Baum-Connes conjecture for abelian groups,
possibly twisted by a cocycle, explains the isomorphisms of
(twisted) $KR$-groups that underlie all T-dualities of torus
orientifold string theories.
\end{abstract}
\keywords{orientifold, $KR$-theory, twisting, T-duality,
  real Baum-Connes conjecture, assembly map, Phillips-Raeburn
  invariant, Hilbert symbol}
\subjclass[2010]{Primary 46L60; Secondary 19K35 19L64  81T30  46L85 19L50.}

\maketitle

\section{Introduction}
\label{sec:intro}

This paper was motivated by joint work with Charles Doran and Stefan
Mendez-Diez \cite{Doran:2013sxa,DMDR}, in which we studied type II
orientifold string theories on circles and $2$-tori. In these
theories, D-brane charges lie in twisted $KR$-groups of $(X,\iota)$,
where $X$ is the spacetime manifold and $\iota$ is the involution on
$X$ defining the orientifold structure. (That D-brane charges for
orientifolds are classified by $KR$-theory was pointed out in 
\cite[\S5.2]{Witten:1998}, \cite{Hori:1999me}, and \cite{Gukov:1999},
but twisting (as defined in
\cite{MoutuouThesis,2011arXiv1110.6836M,2012arXiv1202.2057M} and
\cite{Doran:2013sxa}) may arise due to the $B$-field, 
as in \cite{Witten:1998-02}, and/or the charges of the
$O$-planes, as explained in \cite{Doran:2013sxa}.) 
These orientifold theories were found in \cite{DMDR} to
split up into a number of T-duality groupings, with the theories in
each grouping all related to one another by various T-dualities. The
twisted $KR$-groups attached to each of the theories within a
T-duality grouping were all found to be isomorphic to one another, up
to a degree shift.

One thing that was missing in this previous work was a mathematical
explanation for these twisted $KR$ isomorphisms. The purpose of this
paper is to provide such an explanation. In fact, it turns out that
the isomorphisms of $KR$-groups associated with the T-dualities for
torus orientifolds come from the real Baum-Connes assembly maps for abelian
groups, possibly twisted by a cocycle.   Thus these T-dualities may be
explained mathematically by the fact that the real Baum-Connes
conjecture is valid for these cases \cite{MR2082090,MR2077669}.

Since the appearance of the Baum-Connes conjecture in this context
might seem surprising and unmotivated, we should perhaps explain why
were led to look at the Baum-Connes assembly map.  There are two
justifications. One is that in \cite{MR2116734} and in
\cite{MR2560910}, we found that the $K$-theory isomorphisms associated
to T-duality can often explained via the calculation of $K$-theory of
crossed products by $\bR$ using Connes' ``Thom isomorphism'' theorem
of \cite{MR605351}. This theorem of Connes is now recognized as being
a special case of a proof of the Baum-Connes conjecture with
coefficients (in the case $G=\bR$). While the cases treated in
\cite{MR2116734,MR2560910} deal with complex $K$-theory, duality
involving $KR$-theory appeared in \cite[Theorem 2.5]{MR842428} in the case where
the group involved is free abelian. In that theorem, real Baum-Connes
was shown to give an isomorphism from $KO$-homology of a torus to
$KR$-cohomology of another torus of the same dimension (really the
T-dual torus).  This isomorphism turns out (in 
Proposition \ref{prop:typeIS1} below) to be related to T-duality
between the type I and type IA string theories on the circle.

\section{The main construction and results}
\label{sec:main}

We will be working throughout with Atiyah's $KR$-theory
\cite{MR0206940}. This is the topological $K$-theory (with compact supports)
of Real vector bundles $E$ over Real spaces $(X, \iota)$. A Real space
is just a locally compact Hausdorff space $X$ equipped with a
self-homeomorphism $\iota$ satisfying $\iota^2=\id_X$. A Real vector
bundle $E$ over such a space is a complex vector bundle equipped with
a conjugate-linear vector bundle automorphism of period $2$
compatible with $\iota$. The
$KR$-theory of $(X, \iota)$ can be identified with the topological
$K$-theory of the real Banach algebra $C_0(X,\iota) =\{f\in C_0(X)
\mid f(\iota(x)) = \overline{f(x)}\}$. We shall use 
the indexing convention of \cite{MR1031992,Doran:2013sxa,DMDR}:
$\bR^{p,q}$ denotes $\bR^p\oplus 
\bR^q$ with the involution that is $+1$ on the first summand and $-1$
on the second summand, and $S^{p,q}$ denotes the unit sphere in
$\bR^{p,q}$.

\subsection{Circle orientifolds}
\label{sec:circle}

We begin with the case of orientifolds on a circle, with the
involution coming from a linear involution on $\bR^2$, restricted to
the unit circle. It was found in \cite{Gao:2010ava,Doran:2013sxa,DMDR}
that there are four such orientifold theories, known in the physics
literature as types I, $\tI$, IA, and
$\tIA$. These split into two T-duality groupings, one of
which contains theories I and IA, corresponding to the Real spaces
$S^{2,0}$ and $S^{1,1}$, and the other of which contains theories
$\tI$ and $\tIA$, corresponding to the Real spaces
$S^{0,2}$ and $S^{1,1}_{(+,-)}$. Here the subscript $(+,-)$ in
$S^{1,1}_{(+,-)}$ indicates 
that of the two $O$-planes in $S^{1,1}$ (i.e., fixed points for the
involution), one has been given a plus sign (meaning that the
Chan-Paton bundle there is of real type) and one has been given a
minus sign (meaning that the Chan-Paton bundle there is of quaternionic type).

Let us now see that the twisted $KR$ isomorphisms in these two
T-duality groupings amount to Baum-Connes assembly maps. The first
group consists of theories I and IA, corresponding to the Real spaces
$S^{2,0}$ and $S^{1,1}$. Thus we want an isomorphism $KR^{1-*}(S^{2,0})
\cong KR^{-*}(S^{1,1})$. (The degree shift by $1$ is explained by the
fact that we are applying T-duality in a single circle, and thus going
from a IIB theory to a IIA theory.) Now $S^{2,0}$ is simply $S^1$ with
a trivial involution, which we can identify with the classifying space
$B\bZ$ of the infinite cyclic group $\bZ$. The real Baum-Connes
conjecture holds \cite{MR2082090,MR2077669}
for amenable groups since the complex Baum-Connes
conjecture holds for these groups \cite{MR1821144}. Thus Baum-Connes
gives an assembly isomorphism
\begin{equation}
\mu\co KO_j(B\bZ) \xrightarrow{\cong} KO_j(C^*_{\bR}(\bZ)).
\label{eq:BCZ}
\end{equation}
We can ``unpack'' this as follows.  $B\bZ = S^1 = S^{2,0}$, so the
left-hand side is $KO_j(S^1) \cong KO^{1-j}(S^1)$ by Poincar\'e
duality, since $S^1$ is a spin manifold.  On the other hand,
for an abelian locally compact group $G$,  the Fourier transform sends
$L^1_{\bR}(G)$ (as a $*$-algebra with convolution multiplication) to a
dense subalgebra of 
\[
C_0(\widehat G, \iota)= \left\{ f\in C_0(\widehat G) :
f(x^{-1}) = \overline{f(x)} \text{ for }x\in \widehat G\right\},
\]
where $\widehat G$ is the Pontrjagin dual of $G$ and $\iota$ is the
involution given by group inversion.  Taking $C^*$-completions gives
$C^*_{\bR}(G)\cong C_0(\widehat G, \iota)$. In the case of $G=\bZ$,
$\widehat G=\bT$ and inversion $\iota$ on $\bT$ is complex
conjugation, and thus 
\[
KO_j(C^*_{\bR}(\bZ))\cong KO_j(C(\bT, \iota)) = KO_j(C(S^{1,1})) = KR^{-j}(S^{1,1}).
\]
Putting everything together, we obtain
\begin{proposition}
\label{prop:typeIS1}
There is a natural isomorphism $KO^{1-j}(S^1)\to KR^{-j}(S^{1,1})$
given by the composite $\mu\circ \delta$, where $\mu$ is the real
Baum-Connes assembly map of \eqref{eq:BCZ}
for the discrete group $\bZ$ and where
$\delta\co KO^{1-j}(S^1)\to KO_j(S^1)$ is Poincar\'e duality {\lp}given
analytically by Kasparov product with the class of the Dirac
operator{\rp}. 
\end{proposition}
\begin{proof}
All of this was outlined above.  The statement and proof of
analytical Poincar\'e duality may be
found in \cite[\S4]{MR918241}.
\end{proof}

The more interesting and subtle case comes from the other T-duality
grouping, consisting of the theories of types $\tI$ and
$\tIA$. We want an isomorphism $KR^{1-*}(S^{0,2})
\cong KR^{-*}_{(+,-)}(S^{1,1})$. (Here the $(+,-)$ decoration on the
$KR$-groups was explained in \cite[\S4]{Doran:2013sxa}.
The groups $KR^{-*}_{(+,-)}(S^{1,1})$ are the topological $K$-theory of
a real Banach algebra 
which is locally Morita equivalent to $C(S^{1,1})$ near the fixed
point with the $+$ decoration and locally Morita equivalent to
$C(S^{1,1})\otimes \bH$ near the fixed
point with the $-$ decoration.) Such an isomorphism was obtained
``experimentally'' in \cite[\S4.1]{Doran:2013sxa}, but the treatment
there didn't really explain where this isomorphism comes from (except
in terms of the physics interpretation using T-duality of orientifold
theories). 

Let $G=\left\langle a,b \mid ab=ba,\ b^2=1\right\rangle$. This is an
abelian group isomorphic to $\bZ\times \bZ/2$. On this group we can
define a (normalized)
$2$-cocycle $\omega$ with values in $O(1)=\{\pm 1\}$ by $\omega(a, b)
= \omega(b, a) = \omega(b, b) = -1$, $\omega(a, a)=+1$.
Let $C^*(G,\omega)$ be the (complex) group $C^*$-algebra of $G$
twisted by this cocycle, and let $C^*_{\bR}(G,\omega)$ be the
corresponding real $C^*$-algebra.
\begin{lemma}
\label{lem:S1pm}
With $G$ and $\omega$ as just defined, $C^*(G,\omega)$ has spectrum
$S^1$ and 
\[
KO_j(C^*_{\bR}(G,\omega))\cong KR^{-j}_{(+,-)}(S^{1,1}).
\]
\end{lemma}
\begin{proof}
We begin by noting that $C^*(G,\omega)$ is the universal $C^*$-algebra
on two unitaries $U$ and $V$ satisfying $V^2=-1$ and $UV=-VU$. An
irreducible unitary $\omega$-representation of $G$, when restricted to
$H = \left\langle a\right\rangle$, must be a sum of two unitary
characters $U\mapsto z$ and $U\mapsto -z$, $z\in \bT$, since $V$
conjugates $U$ to $-U$. Then the usual ``Mackey machine'' argument
shows that the representation is induced from one or the other of
these characters, and so the spectrum of the $C^*$-algebra is
naturally identified to the quotient, which is again a circle, 
of $\bT$ by the antipodal map $z\mapsto -z$.

Now the representations $U\mapsto z$, $z\ne 1\text{ or }-1$, 
of $H$ are not defined over $\bR$
and assemble in conjugate pairs to two-dimensional irreducible
representations $U\mapsto \begin{pmatrix}\cos\theta & \sin\theta\\
-\sin\theta& \cos\theta \end{pmatrix}$, $\theta=\arg z$, 
of $H$ over $\bR$. So we see
that the Real involution on the spectrum of $C^*(G,\omega)$ must
interchange $z$ and $\overline z$, giving the Real space $S^{1,1}$.
We need to determine what happens at the fixed points.  At $z=\pm 1$
the associated representation of $G$ is given by
\[
U = \begin{pmatrix}1 & 0\\ 0 & -1 \end{pmatrix},
\quad
V = \begin{pmatrix}0 & 1\\ -1 & 0 \end{pmatrix},
\]
and is defined over $\bR$. Furthermore, these matrices generate
$M_2(\bR)$, which is Morita equivalent to $\bR$.
So this point corresponds to an $O^+$-plane.

When $z=\pm i$ (note that this is a fixed point for complex
conjugation after we divide out by multiplication by $-1$), the 
associated representation of $G$ is given by
\[
U = \begin{pmatrix}i & 0\\ 0 & -i \end{pmatrix},
\quad
V = \begin{pmatrix}0 & 1\\ -1 & 0 \end{pmatrix}.
\]
Then
\[
UV = \begin{pmatrix}0 & i\\ i & 0 \end{pmatrix},
\]
and the $\bR$-span of $I$, $U$, $V$, and $UV$ give a copy of the
quaternions $\bH$. So at this point we get an irreducible real
representation of quaternionic type, and this point corresponds to an
$O^-$-plane.  Now we see that
$C^*_{\bR}(G,\omega)$ is an algebra satisfying the properties of
\cite[Theorem 1]{Doran:2013sxa}, whose topological $K$-theory gives
$KR^{-j}_{(+,-)}(S^{1,1})$. 
\end{proof}

\begin{theorem}
\label{thm:BCZtwist}
The real Baum-Connes isomorphism for the group $G$ defined above,
with twist by the cocycle $\omega$, reduces to
\begin{equation}
\label{eq:BCZtwist}
\mu\co KSC_j\cong
KO_j^{(G, \omega)}(\bR) \xrightarrow{\cong}
KO_j(C^*_{\bR}(G,\omega))\cong KR^{-j}_{(+,-)}(S^{1,1}).
\end{equation}
Upon composition with Poincar\'e duality $KSC^{1-j}\to KSC_j$, this
becomes the isomorphism of twisted $KR$ groups of
\cite[\S4.1]{Doran:2013sxa} that underlies the T-duality between the 
orientifold theories of types $\tI$ and
$\tIA$.
\end{theorem}
\begin{proof}
Note that the central extension of $G$ by $O(1)$ classified by
$\omega$ is $\tG=\left\langle \ta,\tb \mid
\ta\tb=\tb^{-1}\ta,\ \tb^4=1\right\rangle$.  (Here $\ta$ and $\tb$ are
lifts of $a$ and $b$, respectively.) This is a solvable group of the
form $K\rtimes \bZ$, where $K$ is the torsion subgroup, a cyclic
group of order $4$ generated by $\tb$, and the generator $\ta$ of
$\bZ$ acts on $K$ by conjugating $\tb$ to $\tb^3=\tb^{-1}$. So the
left-hand and right-hand sides of the real Baum-Connes map for
$(G,\omega)$ are direct summands in the corresponding sides of the
real Baum-Connes map for $\tG$. Since $\tb^2$ is central of order $2$,
$C^*_{\bR}(\tG)$ splits as $C^*_{\bR}(G,\omega)\oplus C^*_{\bR}(G)$,
where the two summands correspond to representations where $\tb^2$
takes the value $-1$ (resp., $+1$).

The left-hand side of the real Baum-Connes map for $\tG$ is
$KO_j^{\tG}(\bR)$, since $\bR$, with the action where $\ta$ acts by
translation by $1$ and $\tb$ acts trivially, is the universal proper
$\tG$-space. Furthermore, as pointed out by Kasparov \cite[comments
  following Definition 5]{MR1388299}, 
\[
KO_j^{\tG}(\bR) \cong KKO^{-j}(C_0^{\bR}(\bR)\rtimes \tG, \bR),
\]
so we need to understand the structure of the crossed product
\[
C_0^{\bR}(\bR)\rtimes \tG\cong \bigl(C_0^{\bR}(\bR)\rtimes K\bigr)
\rtimes \bZ.
\]
Since $K$ acts trivially on $\bR$, 
\[
C_0^{\bR}(\bR)\rtimes K\cong
C_0^{\bR}(\bR)\otimes C^*_{\bR}(K)\cong C_0^{\bR}(\bR)\otimes
\bigl(\bR\oplus\bR\oplus \bC\bigr).
\]
The two summands of $\bR$ in $C^*_{\bR}(K)$ correspond to the
representations $\tb\mapsto \pm 1$, which are trivial on $\tb^2$.  So
these summands will go to $C^*_{\bR}(G)$ on the right-hand side.  The
summand sent under $\mu$ to $C^*_{\bR}(G,\omega)$ is thus:
\[
\bigl(C_0^{\bR}(\bR)\otimes \bC\bigr) \rtimes \bZ.
\]
But we have to be careful with the action of $\bZ$. On
$C_0^{\bR}(\bR)$, the generator $\ta$ of $\bZ$ acts by translation by
$1$.  But on $\bC$, which corresponds to the representations
$\tb\mapsto \pm i$ of $K$, $\ta$ acts by complex conjugation since
conjugation by $\ta$ sends $\tb$ to $\tb^{-1}$. Thus for $f\in
C_0(\bR)$, $\ta\cdot f(x) = \overline{f(x+1)}$. The crossed product
$C_0(\bR)\rtimes \bZ$ is therefore Morita equivalent to 
\[
\{f\in C(\bR/2\bZ)\mid f(x+1)= \overline{f(x)}\}
\cong C(S^{0,2}),
\]
and $\mu$ reduces to an isomorphism
\[
KR_j(S^{0,2})\to KR^{-j}_{(+,-)}(S^{1,1}).
\]
The left-hand side is isomorphic to $KR^{1-j}(S^{0,2})= KSC^{1-j}$ via
Poincar\'e duality for $S^{0,2}$, and the theorem follows via real
Baum-Connes for solvable groups \cite{MR2082090,MR2077669} and Lemma
\ref{lem:S1pm}. 
\end{proof}

\subsection{$2$-Torus orientifolds}
\label{sec:torus}

In \cite{Gao:2010ava,DMDR}, a study was made of all possible type II
orientifold string theories on $2$-tori. For reasons of supersymmetry,
it was assumed that the $2$-torus $X$ on which spacetime was compactified
is equipped with a complex structure making it into an elliptic curve,
and that the orientifold involution $\iota$ on $X$ is holomorphic in
type IIB and antiholomorphic in type IIA.  It turned out that there
are ten distinct such theories, divided into 3 groupings.  All the
theories in a single grouping are related to one another by sequences
of T-dualities.  Some of these T-dualities had been predicted earlier,
for example in \cite{MR1008294,Witten:1998-02}. 

The first group
contains the type I theory on $T^2$ (this is the IIB orientifold for
the Real space $(T^2, \id_{T^2})$), the IIA orientifold theory on  $T^2$ for an
orientation-reversing involution with fixed set $S^1\amalg S^1$, and
the IIB orientifold theory on  $T^2$ for an orientation-preserving
involution with $4$-point fixed set, each point with positive
$O$-plane charge. The associated Real spaces are $S^{2,0} \times
S^{2,0}$, $S^{1,1} \times S^{2,0}$, and $S^{1,1} \times S^{1,1}$. The
T-dualities within this group can all be easily obtained (by taking products)
from the case of types I and IA on the circle (see Proposition
\ref{prop:typeIS1}), so we will not discuss them further.

The second group consists of the $\tI$ theory (the IIB orientifold on
$S^{2,0}\times S^{0,2}$), the $\tIA$ theory on $S^{1,1}\times S^{2,0}$,
the IIA theory on $S^{1,1}\times S^{0,2}$, and the IIB theory on
$T^2$ for an orientation-preserving 
involution with $4$-point fixed set, where half the fixed points have positive
$O$-plane charge and half have negative $O$-plane charge. The
T-dualities within this group can all be easily obtained (by taking products)
from the case of types $\tI$ and $\tIA$ on the circle (see Theorem
\ref{thm:BCZtwist}), along with the I--IA duality on the circle
(Proposition \ref{prop:typeIS1}), so again we will not discuss them further.

The interesting and subtle case involves the final T-duality
grouping.  There are three theories in this group, the type I theory
with non-trivial $B$-field (called the theory ``without vector
structure'' in \cite{Witten:1998-02}), the IIA theory for an
antiholomorphic involution with fixed set $S^1$ and quotient space a
M\"obius strip, and the IIB orientifold theory for a holomorphic
involution with four fixed points, three of which have positive
$O$-plane charge and one of which has negative $O$-plane charge.
In \cite{Doran:2013sxa,DMDR}, the twisted $KR$-groups for all of these
theories were computed and found to be (abstractly) isomorphic, but no
explicit isomorphisms of twisted $KR$-groups were obtained. We will
now remedy this deficiency.

Let us now switch notation from Section \ref{sec:circle} and let $G$
be a free abelian group on two generators $a$ and $b$, and let
$\nu\in Z^2(G, O(1))$ be the normalized $2$-cocycle on $G$ with $\nu(a,
b) = \nu(b, a) = -1$, $\nu(a, a) = \nu(b, b) = +1$. The
associated central extension of $G$ by $O(1)$ is
$\tG=\left\langle \ta,\tb, c \mid
\ta\tb=c\tb\ta,\ c^2=1\right\rangle$.  (Here $\ta$ and $\tb$ are
lifts of $a$ and $b$, respectively, and $c$, corresponding to $-1\in
O(1)$, is central.) The counterpart to Lemma \ref{lem:S1pm} is the
following:
\begin{lemma}
\label{lem:T2pm}
With $G$ and $\nu$ as just defined, $C^*(G,\nu)$ has spectrum
$T^2$ and 
\[
KO_j(C^*_{\bR}(G,\nu))\cong KR^{-j}_{(+,+,+,-)}(S^{1,1}\times S^{1,1}).
\]
\end{lemma}
\begin{proof}
We begin by computing the complex $C^*$-algebra $C^*(G,\nu)$. This
is the free $C^*$-algebra on two unitary generators $U$ and $V$
satisfying the commutation rule $UV=-VU$. In other words, it is just
the (rational) noncommutative torus $A_{1/2}$. This has spectrum $T^2$
and is the algebra of sections of a stably trivial, locally trivial bundle of
algebras over $T^2$, with fibers isomorphic to $M_2(\bC)$
(see for instance \cite{MR636190,zbMATH03857768}). The center
of $C^*(G,\nu)$ is generated by $U^2$ and $V^2$, which together
generate a copy of $C(T^2)$. If $z,w\in \bT$, then in an irreducible
representation with $U^2\mapsto z$ and $V^2\mapsto w$, the eigenvalues
of $U$ are $\pm z^{1/2}$ and the eigenvalues of $V$ are $\pm
w^{1/2}$. When either $z^{1/2}$ or $w^{1/2}$ is non-real, to get a
real representation of $C^*_{\bR}(G,\nu)$ we need to take the
eigenvalues of $U$ to be $\pm z^{1/2}, \pm \overline z^{1/2}$ 
and the eigenvalues of $V$ to be $\pm w^{1/2},\pm \overline
w^{1/2}$. So from this analysis, we can see that the underlying Real
space of $C^*_{\bR}(G,\nu)$ must be $S^{1,1}\times S^{1,1}$. It
remains to compute the $O$-plane charges at the fixed points, where
$z=\pm 1$ and $w=\pm 1$. If $z=1$ and $w=1$, we have an irreducible
representation given by
\[
U = \begin{pmatrix}1 & 0\\ 0 & -1 \end{pmatrix},
\quad
V = \begin{pmatrix}0 & 1\\ 1 & 0 \end{pmatrix},
\]
and this is defined over $\bR$. Furthermore, these matrices generate
$M_2(\bR)$, which is Morita equivalent to $\bR$.
So this point corresponds to an $O^+$-plane.  If $z=1$ and $w=-1$ (or
the other way around---these cases are symmetrical), then we have an
irreducible representation given by 
\[
U = \begin{pmatrix}1 & 0\\ 0 & -1 \end{pmatrix},
\quad
V = \begin{pmatrix}0 & 1\\ -1 & 0 \end{pmatrix},
\]
and this is defined over $\bR$. Furthermore, these matrices generate
$M_2(\bR)$, which is Morita equivalent to $\bR$.
So again these points correspond to $O^+$-planes. But if $z=w=-1$,
then just as in the proof of Lemma \ref{lem:S1pm}, we have an
irreducible representation generated by
\[
U = \begin{pmatrix}i & 0\\ 0 & -i \end{pmatrix},
\quad
V = \begin{pmatrix}0 & 1\\ -1 & 0 \end{pmatrix}.
\]
These matrices generate a copy of the quaternions $\bH$, and we have
an $O^-$-plane. So $C^*_{\bR}(G,\nu)$ is a real continuous-trace
algebra with spectrum $S^{1,1}\times S^{1,1}$ and with three
$O^+$-planes, one $O^-$-plane. This algebra satisfies all the
properties of \cite[Theorem 1]{Doran:2013sxa}, and its topological
$K$-theory gives $KR^{-j}_{(+,+,+,-)}(S^{1,1}\times S^{1,1})$.
\end{proof}
\begin{remark}
The reader used to the theory of the Hilbert symbol will note that the
sign choice $(+,+,+,-)$ arose here from the fact that if we compute
the Hilbert symbol $(z,w)$ for all choices $z,w\in \{\pm 1\}$, then it
takes the value $+1$ when either $z$ or $w$ is positive and takes the
value $-1$ exactly when $z=w=-1$.
\end{remark}
\begin{theorem}
\label{thm:BCZ2twist}
The real Baum-Connes isomorphism for the group $G$ defined above,
with twist by the cocycle $\nu$, reduces to
\begin{equation}
\label{eq:BCZ2twist}
\mu\co KO_j(T^2, w)\cong
KO_j^{(G, \nu)}(\bR^2) \xrightarrow{\cong} KO_j(C^*_{\bR}(G,\nu))\cong 
KR^{-j}_{(+,+,+,-)}(S^{1,1}\times S^{1,1}).
\end{equation}
Here $w$ is the nontrivial element of $H^2(T^2,\bZ/2)$ and $KO_*(T^2,
w)$ denotes the associated twisted $KO$-homology.
Upon composition with Poincar\'e duality 
\[KO^{2-j}(T^2, w) \to KO_j(T^2, w),
\] 
this becomes the isomorphism of twisted $KR$ groups of
\cite[\S4.2 and \S5]{Doran:2013sxa} that underlies the T-duality between the 
``type I theory without vector structure'' and the IIB
orientifold theory on $T^2$ with four fixed points, three of them
$O^+$-planes and one of them an $O^-$-plane.
\end{theorem}
\begin{proof}
Once again, we use the fact that the left-hand side of the real
Baum-Connes conjecture for $(G,\nu)$ is a direct summand in the
left-hand side of the real Baum-Connes conjecture for $\tG$.
The universal proper $\tG$-space is $\bR^2$, with $\ta$ acting by
translation by $(1,0)$ and $\tb$ acting by
translation by $(0,1)$. So, again following \cite[comments
  following Definition 5]{MR1388299}, 
\[
KO_j^{\tG}(\bR^2) \cong KKO^{-j}(C_0^{\bR}(\bR^2)\rtimes \tG, \bR),
\]
and we need to understand the structure of the crossed product
\[
C_0^{\bR}(\bR^2)\rtimes \tG\cong 
\bigl(C_0^{\bR}(\bR^2)\rtimes G\bigr) \oplus
\bigl(C_0^{\bR}(\bR^2)\rtimes_\nu G\bigr).
\]
Here the first summand $C_0^{\bR}(\bR^2)\rtimes G$, where $c\in \tG$
acts by $+1$,  corresponds to $C^*_{\bR}(G)$
on the right-hand side of the Baum-Connes conjecture, and the other
summand (the one we are interested in, where $c$ acts by $-1$)
corresponds to $C^*_{\bR}(G,\nu)$. This summand
$C_0^{\bR}(\bR^2)\rtimes_\nu G$ 
has complexification Morita equivalent to $C(T^2)$, since $G$ acts
freely on $\bR^2$ with quotient $T^2$, and since the complex
Dixmier-Douady class has to vanish since $H^3(T^2,\bZ)=0$.
So $C_0^{\bR}(\bR^2)\rtimes_\nu G$ is a real continuous-trace
algebra with spectrum $T^2$, with trivial involution.

We still need to compute the \emph{real} Dixmier-Douady class. 
The calculation is the exact analogue of a result of
Wassermann \cite[Theorem 5]{MR2631479} and
Raeburn-Williams \cite[Theorem 4.1]{MR768739}
(see also \cite[Remark on pp.\ 26--27]{MR920145}) in the complex case,
except that we have to replace $U(1)=\bT$ by $O(1)$.
The result says that this class $w\in H^2(\bR^2/\bZ^2, \bZ/2)$ is
exactly the image of the class of $\nu$ in $H^2(BG, O(1))$, so it's
the nontrivial element of $H^2(T^2, \bZ/2)$.  Thus, via an application
of Lemma \ref{lem:T2pm}, the real
Baum-Connes isomorphism for $(G, \nu)$ reduces to
\[
KO_j(T^2, w) \xrightarrow{\cong} KO_j(C^*_{\bR}(G,\nu))\cong
KR^{-j}_{(+,+,+,-)}(S^{1,1}\times S^{1,1}). 
\]
Composing with Poincar\'e duality, we get the desired isomorphism
\[
KO^{2-j}(T^2, w) \xrightarrow{\cong} KR^{-j}_{(+,+,+,-)}(S^{1,1}\times
S^{1,1}). 
\]
The degree shift by $2$ was explained in \cite{DMDR} by the fact that
the associated type IIB orientifold string theories differ by a double
T-duality. 
\end{proof}

\subsection{The real Phillips-Raeburn obstruction}
\label{sec:PR}

Before we get to the last case, we need a short digression on the real
Phillips-Raeburn obstruction, which might be known to some people in
the noncommutative geometry community but doesn't seem to be
discussed in the literature. For simplicity, we will just deal with
the case of automorphisms of $C_0(X,\cK_{\bR})$, where $X$ is a second
countable locally compact Hausdorff space, which in our case of
interest will have the homotopy type of a finite CW complex,
and $\cK_{\bR}$ denotes the
compact operators on an infinite-dimensional separable real Hilbert
space. Since every automorphism of $\cK_{\bR}$ is inner, $\Aut
\cK_{\bR}\cong PO=O/\{\pm 1\}$, the infinite-dimensional projective
orthogonal group.  Here $O$ is given the strong (or weak) operator
topology and is contractible, for the same reason as the infinite
unitary group $U$ \cite[Lemme 3]{MR0163182}.

In the complex case, recall \cite[\S5.4]{MR1634408} that
$\Aut_X C_0(X,\cK)$, the spectrum-fixing automorphisms, consists of
\emph{locally inner} automorphisms, and that the obstruction to a
locally inner automorphism being inner is given by the
Phillips-Raeburn obstruction in $H^2(X, \bZ)$ (\cite{MR589649},
\cite[Theorem 5.42]{MR1634408}). This obstruction is easy to describe:
a locally inner automorphism $\theta$ of $\Aut_X C_0(X,\cK)$ can be
viewed as a continuous map $\theta\co X\to \Aut \cK\cong PU$, and
since the projective unitary group is a classifying space for $\bT$,
and thus has the homotopy type of a $K(\bZ,2)$ space,
the homotopy class $[\theta]$ of $\theta$ gives a class in $H^2(X, \bZ)$.
Vanishing of this class is the obstruction to lifting the map $\theta$
to a map $X\to U$, and thus to $\theta$ being inner.

There is another way of understanding the Phillips-Raeburn
obstruction \cite{MR749160}, which is also relevant. An automorphism
$\theta$ is the 
same thing as an action of $\bZ$, and if $\theta$ is spectrum-fixing,
$C_0(X,\cK)\rtimes_\theta \bZ$ is a continuous-trace algebra with
spectrum $Y$ which is a principal $S^1$-bundle over $X$, the 
$S^1$-action coming from the dual action of $\bT=\widehat{\bZ}$ on the
crossed product. The class
$[\theta]\in H^2(X, \bZ)$ is precisely the Chern class of this bundle.

Now let's consider the real case. The proof that every spectrum-fixing
automorphism of $C_0(X,\cK_{\bR})$ is locally inner is exactly the
same as in the complex case. A locally inner automorphism $\theta$ of
$\Aut_X C_0(X,\cK_{\bR})$ can be 
viewed as a continuous map $\theta\co X\to \Aut \cK_{\bR}\cong PO$, and
since the projective orthogonal group is a classifying space for $O(1)$,
and thus has the homotopy type of a $K(\bZ/2,1)$ space,
the homotopy class $[\theta]$ of $\theta$ gives a class in $H^1(X, \bZ/2)$.
Vanishing of this class is the obstruction to lifting the map $\theta$
to a map $X\to O$, and thus to $\theta$ being inner.

As in the complex case, we can also describe the real Phillips-Raeburn
obstruction as the class of a bundle obtained from the crossed
product. Since the map $PO\to PU$ is null-homotopic, if we complexify,
a locally inner automorphism $\theta$ becomes inner and 
$C_0(X,\cK)\rtimes_{\theta_\bC} \bZ$ is isomorphic to $C_0(X,\cK)\otimes
C^*(\bZ) \cong C_0(X\times S^1,\cK)$. In particular, the associated
$S^1$-bundle $X\times S^1\to X$ is trivial.  However, the triviality
comes from the fact that we have forgotten the real structure. The
spectrum of $C_0(X,\cK_{\bR})\rtimes_\theta \bZ$ is a Real space, and
over an open set $U\subseteq X$ where $\theta$ is inner,
this Real space is $U\times S^{1,1}$, since $C^*_{\bR}(\bZ)\cong
C(S^{1,1})$. Globally, $\bigl(C_0(X,\cK_{\bR})\rtimes_\theta
\bZ\bigr)\,\widehat{}\,$ is a bundle $Y\to X$ of Real spaces over $X$, with
fibers $S^{1,1}$, which is trivial after forgetting the Real
structure. The fixed set for the involution on the total space of the
bundle is thus a principal $\bZ/2$-bundle over $X$, i.e., a regular
double cover, and the real Phillips-Raeburn obstruction is the 
characteristic class of this cover in $H^1(X, \bZ/2)$.
To see this, observe that if $\theta$ is inner, then the
crossed product by $\theta$ is isomorphic to $C_0(X,\cK_{\bR})\otimes
C^*_{\bR}(\bZ)$, which has spectrum $X\times S^{1,1}$ (as a Real
space), and the bundle is trivial. In the other direction, suppose the
$\bZ/2$-bundle over $X$ is trivial. That means there is a global
Real section of $Y\to X$, which we can identify with a continuous
family of real representations of the crossed product. Since we
already know that $\theta_{\bC}$ is inner, $\theta$ is implemented by
some strongly continuous $u\co X\to U(\cH_\bR \otimes \bC)$ which is orthogonal
modulo center, i.e., by some $u\co X\to O\cdot \bT$. Since we have a
global family of real representations, we can take $u$ to be pointwise
orthogonal, and thus the real Phillips-Raeburn obstruction vanishes.
So vanishing of the real Phillips-Raeburn obstruction in
$H^1(X,\bZ/2)$ is necessary and sufficient for vanishing of the
$\bZ/2$-bundle over $X$, and so it must correspond precisely to the
usual characteristic class of this bundle, which is the obstruction to
triviality of the bundle of spectra in the category of Real spaces.

\subsection{The species 1 IIA theory}
\label{sec:spec1}

There is one case left to handle, the one which in many respects is
the most subtle.  This is the T-duality between the type IIA
orientifold theory associated to a species $1$ 
real elliptic curve\footnote{For a smooth irreducible 
  projective curve $X$ defined
  over $\bR$, the set of complex points is a connected compact Riemann
  surface and the set of real points is a topologically a finite
  disjoint union of circles. The \emph{species} is the number of
  components of the set of real points.}
and the type I theory twisted by a nontrivial $B$-field.

The species $1$ IIA theory, from the point of view of $KR$-theory,
corresponds to a Real space $(T^2,\iota)$ in the sense of Atiyah, where the
underlying topological space is $T^2$, and the involution $\iota$ is
orientation-reversing, with fixed set $S^1$ and quotient space a
M{\"o}bius strip. As explained in
\cite[\S2 and \S5.2.3]{DMDR}, $(T^2, \iota)$ has a concrete
realization as $(\bC/\Lambda, z\mapsto \bar z)$, where $\Lambda$ is
the lattice in $\bC$ generated by $1$ and by $\tau=\frac12 + i\tau_2$,
where $\tau_2>0$. What makes this case tricky is that $\Re\tau =
\frac12$, so the lattice $\Lambda$ is \emph{not} rectangular. Note,
however, that we have a sublattice $\Lambda'\subset \Lambda$ of index
$2$, where $\Lambda'$ is generated by $1$ and $2\tau$, or equivalently
by $1$ and $2i\tau_2$. So $\Lambda'$ \emph{is} rectangular. The
involution $z\mapsto \bar z$ on $\bC/\Lambda'$ gives the Real space
$S^{2,0}\times S^{1,1}$ studied above in Section \ref{sec:torus}. So
$(T^2,\iota)$ has $S^{2,0}\times S^{1,1}$ as a double cover.

The desired T-duality will involve real Baum-Connes with coefficients
for the group $\bZ$. This conjecture is stated in \cite[\S9]{MR1292018}
(with the substitution of $KO$ for $KU$), and is in fact a theorem in
the case $G=\bZ$, by \cite{MR2077669}, for example. First we need to explain
what the coefficient $\bZ$-algebra is.

Let $A = C(S^1, M_2(\bR))$, which as a real $C^*$-algebra is obviously
Morita equivalent to $C^{\bR}(S^1)$, associated to the Real space
$S^{2,0}$. We now equip this with an interesting action of $\bZ$ as
follows. Identify $z=e^{2\pi i t}\in \bT$ with the rotation matrix
\[
r(z)= \begin{pmatrix}\cos 2\pi t & \sin 2\pi t\\ -\sin 2\pi t & \cos
  2\pi t 
\end{pmatrix} \in SO(2),
\]
and define $\theta\in \Aut A$ to be given by
\[
(\theta f)(z) = r(z^{1/2})f(z)r(z^{1/2})^{-1}.
\]
At first sight this seems ambiguous since $z\in \bT$ has two square
roots, $\pm e^{\pi i t}$, but since they differ by $-1\in Z(O(2))$, the inner
automorphisms they define are the same and so $\theta$ is well
defined, regardless of what choice one makes for the square root.

The interesting feature of the automorphism $\theta$ is that it is
\emph{locally inner} but not inner. From the discussion in Section
\ref{sec:PR} above, it should be clear that $\theta$ has nontrivial
real Phillips-Raeburn class in $H^1(S^1,\bZ/2)$.

\begin{lemma}
\label{lem:AZ}
With $A$ and the automorphism $\theta$ as just defined, the crossed
product $A\rtimes_\theta \bZ$ is strongly Morita equivalent to
the commutative real $C^*$-algebra associated to the Real space $(T^2,
\iota)$, where $\iota$ is the orientation-reversing involution
associated to the species $1$ IIA elliptic curve orientifold theory.
\end{lemma}
\begin{proof}
We take the crossed product in stages.  Since the Phillips-Raeburn
obstruction of $\theta$ is $2$-torsion, the Phillips-Raeburn
obstruction of $\theta^2$ vanishes, the subgroup $2\bZ\subset \bZ$
acts on $A$ by inner automorphisms, and $A\rtimes 2\bZ$ is Morita
equivalent to $A\otimes C^*_{\bR}(2\bZ)$, a stably commutative real $C^*$-algebra
corresponding to the Real space $S^{2,0}\times S^{1,1}$. The
crossed product $A\rtimes_\theta  \bZ$ is then obtained 
(up to strong Morita equivalence) by taking a further
crossed product by the quotient group $\bZ/2\bZ$.  (Use the
Packer-Raeburn trick of \cite{MR1002543}.) We will see shortly that
taking this further crossed product amounts to dividing $S^{2,0}\times
S^{1,1}$ by a free involution.

From the discussion in Section \ref{sec:PR} above, 
the crossed product $A\rtimes_\theta \bZ$ is Morita equivalent to an
abelian real $C^*$-algebra, associated to a Real space $(X,\iota)$.
This space  $(X,\iota)$ has as double cover the Real space associated
to  $A\rtimes_\theta 2\bZ$, or $S^{2,0}\times S^{1,1}$. But by
nontriviality of the Phillips-Raeburn invariant of $\theta$, the natural
projection map $(X,\iota)\to S^{2,0}$ does not have a Real
section. However, we know that just as a bundle of topological spaces
(i.e., forgetting the Real structure), $X$ is a trivial $S^1$-bundle
over $S^1$, and is thus a $2$-torus.  Furthermore,
the inclusion $2\bZ\subset \bZ$ induces a map $C^*(2\bZ)\to C^*(\bZ)$
or $C(S^1)\to C(S^1)$ which is dual to the covering map $z\mapsto z^2$
on $\bT$. So the double 
covering $S^{2,0}\times S^{1,1} \to (X,\iota)$ is nontrivial on
the $S^{1,1}$ factor. But it is also nontrivial on the $S^{2,0}$
factor, because of the nontriviality of the Phillips-Raeburn
invariant. From our previous description of the 
covering map $\bC/\Lambda'\to \bC/\Lambda$, we recognize $(X,\iota)$
as being associated to the species $1$ real elliptic curve.
\end{proof}
\begin{theorem}
The isomorphism $KO^{1-j}(T^2, w)\to KO^{-j}(X,\iota)$ associated to
the T-duality between the ``type I theory without vector structure''
and the IIA orientifold associated to a species $1$ real elliptic
curve can be obtained as the composition
\begin{multline*}
KO^{1-j}(T^2, w)\xrightarrow[\cong]
{\text{\textup{Theorem \ref{thm:BCZ2twist}}}}
KKO_{j+1}(C_0^{\bR}(\bR^2)\rtimes_\nu\bZ^2,\bR)\\
\xrightarrow[\cong]{} 
KKO_{j+1}\left(\left(C_0(\bR)\otimes A\right)\rtimes \bZ,\bR\right)
\xrightarrow[\cong]{\text{\textup{Poinc.\ d.}}}
KKO^{\bZ}_j(C_0(\bR), A) \\
\xrightarrow[\cong]{\text{\textup{Baum-Connes}}}
KO_j(A\rtimes_\theta \bZ) \xrightarrow[\cong]{\text{\textup{Lemma
    \ref{lem:AZ}}}} KR^{-j}(T^2, \iota).
\end{multline*}
\label{thm:spec1}
\end{theorem}
\begin{proof}
The real Baum-Connes isomorphism for $\bZ$ with coefficients in $(A,\theta)$
sends $KKO^{\bZ}_j(C_0(\bR), A)$ to $KO_j(A\rtimes_\theta \bZ)$, which
by Lemma 3 can be identified with $KR^{-j}(T^2, \iota)$.  So let's
analyze the left-hand side, $KKO^{\bZ}_j(C_0(\bR), A)$. Since $A$ is
Morita equivalent to $C^{\bR}(S^1)$ and the $\bZ$-action $\theta$ is
the identity on $S^1$, this is isomorphic by Poincar\'e duality to
$KKO^{\bZ}_{j+1}(C_0(\bR)\otimes A, \bR)$ (\cite[Theorem 1]{MR1388299} and 
\cite[Theorem 4.10]{MR918241} --- the degree shift by $1$ comes from
the fact that $\dim S^1=1$), which in turn is isomorphic to 
$KKO_{j+1}\left(\left(C_0(\bR)\otimes A\right)\rtimes \bZ, \bR\right)$.
Here $\bZ$ is acting on $C_0(\bR)$ by translations and on $A$ by
$\theta$. On the other hand, we know by the proof of Theorem
\ref{thm:BCZ2twist} that 
\begin{multline}
KO^{1-j}(T^2, w)\xrightarrow[\cong]{\text{\textup{Poincar\'e
      duality}}}
KO_{j+1}(T^2, w)
\cong \\ KKO_{j+1}\bigl(C_0^{\bR}(\bR^2) \rtimes_\nu
 \bZ^2, \bR\bigr)
\cong KKO_{j+1}\left(\left(C_0^{\bR}(\bR)\otimes A\right)\rtimes
\bZ, \bR\right).
\label{eq:Gtw}
\end{multline}
The last isomorphism in \eqref{eq:Gtw}
is obtained by decomposing the twisted crossed
product by $\bZ^2$ as a crossed product by $\bZ$, which gives
something Morita equivalent to $C_0^{\bR}(\bR)\otimes A$, followed by
another crossed product by $\bZ$. This second crossed product has
$\bZ$ acting on $\bR$ by translations. The action on $A$ is trivial on
the spectrum $S^1\cong \bR/\bZ$, but not inner
because of the noncommutativity of $\tG$.  So up to inner
automorphisms it is given by $\theta$, which represents the unique
nontrivial Phillips-Raeburn class.  Splicing all the isomorphisms
together, the result follows. 
\end{proof}

\bibliographystyle{hplain}
\bibliography{BCT}

\def\cprime{$'$}
\begin{thebibliography}{10}

\bibitem{MR0206940}
M.~F. Atiyah.
\newblock {$K$}-theory and reality.
\newblock {\em Quart. J. Math. Oxford Ser. (2)}, 17:367--386, 1966.

\bibitem{MR1292018}
Paul Baum, Alain Connes, and Nigel Higson.
\newblock Classifying space for proper actions and {$K$}-theory of group
  {$C^\ast$}-algebras.
\newblock In {\em {$C^\ast$}-algebras: 1943--1993 ({S}an {A}ntonio, {TX},
  1993)}, volume 167 of {\em Contemp. Math.}, pages 240--291. Amer. Math. Soc.,
  Providence, RI, 1994.

\bibitem{MR2082090}
Paul Baum and Max Karoubi.
\newblock On the {B}aum-{C}onnes conjecture in the real case.
\newblock {\em Q. J. Math.}, 55(3):231--235, 2004.

\bibitem{MR605351}
A.~Connes.
\newblock An analogue of the {T}hom isomorphism for crossed products of a
  {$C^{\ast} $}-algebra by an action of {$\mathbb R$}.
\newblock {\em Adv. in Math.}, 39(1):31--55, 1981.

\bibitem{zbMATH03857768}
Marc {De Brabanter}.
\newblock {The classification of rotational rotation $C\sp*$-algebras.}
\newblock {\em {Arch. Math.}}, 43:79--83, 1984.

\bibitem{MR0163182}
Jacques Dixmier and Adrien Douady.
\newblock Champs continus d'espaces hilbertiens et de {$C^{\ast} $}-alg\`ebres.
\newblock {\em Bull. Soc. Math. France}, 91:227--284, 1963.
\newblock available at {\texttt{www.numdam.org}}.

\bibitem{Doran:2013sxa}
Charles Doran, Stefan Mendez-Diez, and Jonathan Rosenberg.
\newblock T-duality for orientifolds and twisted {$KR$}-theory.
\newblock {\em Lett. Math. Phys.}, 104(11):1333--1364, 2014, arXiv:1306.1779.

\bibitem{DMDR}
Charles Doran, Stefan Mendez-Diez, and Jonathan Rosenberg.
\newblock String theory on elliptic curve orientifolds and {$KR$}-theory.
\newblock {\em Comm. Math. Phys.}, to appear, arXiv:1402.4885.

\bibitem{Gao:2010ava}
Dongfeng Gao and Kentaro Hori.
\newblock {On the structure of the {C}han-{P}aton factors for {D}-branes in
  type {II} orientifolds}.
\newblock preprint, 2010, arXiv:1004.3972.

\bibitem{Gukov:1999}
Sergei Gukov.
\newblock {{$K$}-theory, reality, and orientifolds}.
\newblock {\em Comm. Math. Phys.}, 210:621--639, 2000, arXiv:hep-th/9901042.

\bibitem{MR1821144}
Nigel Higson and Gennadi Kasparov.
\newblock {$E$}-theory and {$KK$}-theory for groups which act properly and
  isometrically on {H}ilbert space.
\newblock {\em Invent. Math.}, 144(1):23--74, 2001.

\bibitem{MR636190}
Raphael H{\o}egh-Krohn and Tor Skjelbred.
\newblock Classification of {$C^{\ast} $}-algebras admitting ergodic actions of
  the two-dimensional torus.
\newblock {\em J. Reine Angew. Math.}, 328:1--8, 1981.

\bibitem{Hori:1999me}
Kentaro Hori.
\newblock {D-branes, T duality, and index theory}.
\newblock {\em Adv. Theor. Math. Phys.}, 3:281--342, 1999, hep-th/9902102.

\bibitem{MR918241}
G.~G. Kasparov.
\newblock Equivariant {$KK$}-theory and the {N}ovikov conjecture.
\newblock {\em Invent. Math.}, 91(1):147--201, 1988.

\bibitem{MR1388299}
G.~G. Kasparov.
\newblock {$K$}-theory, group {$C^*$}-algebras, and higher signatures
  (conspectus).
\newblock In {\em Novikov conjectures, index theorems and rigidity, {V}ol.\ 1
  ({O}berwolfach, 1993)}, volume 226 of {\em London Math. Soc. Lecture Note
  Ser.}, pages 101--146. Cambridge Univ. Press, Cambridge, 1995.

\bibitem{MR1031992}
H.~Blaine Lawson, Jr. and Marie-Louise Michelsohn.
\newblock {\em Spin geometry}, volume~38 of {\em Princeton Mathematical
  Series}.
\newblock Princeton University Press, Princeton, NJ, 1989.

\bibitem{MR2116734}
Varghese Mathai and Jonathan Rosenberg.
\newblock {$T$}-duality for torus bundles with {$H$}-fluxes via noncommutative
  topology.
\newblock {\em Comm. Math. Phys.}, 253(3):705--721, 2005, arXiv:hep-th/0401168.

\bibitem{2011arXiv1110.6836M}
E.-K.~M. {Moutuou}.
\newblock Twistings of {$KR$} for {R}eal groupoids.
\newblock 2011, arXiv:1110.6836.

\bibitem{MoutuouThesis}
E.-K.~M. {Moutuou}.
\newblock {\em Twisted groupoid {$KR$}-theory}.
\newblock PhD thesis, Universit{\'e} de {L}orraine, 2012.
\newblock available at {\texttt{http://www.theses.fr/2012LORR0042}}.

\bibitem{2012arXiv1202.2057M}
E.-K.~M. {Moutuou}.
\newblock Graded {B}rauer groups of a groupoid with involution.
\newblock {\em J. Funct. Anal.}, 266(5):2689--2739, 2014, arXiv:1202.2057.

\bibitem{MR1002543}
Judith~A. Packer and Iain Raeburn.
\newblock Twisted crossed products of {$C^*$}-algebras.
\newblock {\em Math. Proc. Cambridge Philos. Soc.}, 106(2):293--311, 1989.

\bibitem{MR589649}
John Phillips and Iain Raeburn.
\newblock Automorphisms of {$C^{\ast} $}-algebras and second \v {C}ech
  cohomology.
\newblock {\em Indiana Univ. Math. J.}, 29(6):799--822, 1980.

\bibitem{MR749160}
John Phillips and Iain Raeburn.
\newblock Crossed products by locally unitary automorphism groups and principal
  bundles.
\newblock {\em J. Operator Theory}, 11(2):215--241, 1984.

\bibitem{MR920145}
Iain Raeburn and Jonathan Rosenberg.
\newblock Crossed products of continuous-trace {$C^\ast$}-algebras by smooth
  actions.
\newblock {\em Trans. Amer. Math. Soc.}, 305(1):1--45, 1988.

\bibitem{MR768739}
Iain Raeburn and Dana~P. Williams.
\newblock Pull-backs of {$C^\ast$}-algebras and crossed products by certain
  diagonal actions.
\newblock {\em Trans. Amer. Math. Soc.}, 287(2):755--777, 1985.

\bibitem{MR1634408}
Iain Raeburn and Dana~P. Williams.
\newblock {\em Morita equivalence and continuous-trace {$C^*$}-algebras},
  volume~60 of {\em Mathematical Surveys and Monographs}.
\newblock American Mathematical Society, Providence, RI, 1998.

\bibitem{MR842428}
Jonathan Rosenberg.
\newblock {$C^\ast$}-algebras, positive scalar curvature, and the {N}ovikov
  conjecture. {III}.
\newblock {\em Topology}, 25(3):319--336, 1986.

\bibitem{MR2560910}
Jonathan Rosenberg.
\newblock {\em Topology, {$C^\ast$}-algebras, and string duality}, volume 111
  of {\em CBMS Regional Conference Series in Mathematics}.
\newblock Published for the Conference Board of the Mathematical Sciences,
  Washington, DC; by the American Mathematical Society, Providence, RI, 2009.

\bibitem{MR1008294}
Augusto Sagnotti.
\newblock Open strings and their symmetry groups.
\newblock In {\em Nonperturbative quantum field theory ({C}arg\`ese, 1987)},
  volume 185 of {\em NATO Adv. Sci. Inst. Ser. B Phys.}, pages 521--528.
  Plenum, New York, 1988, arXiv:hep-th/0208020.

\bibitem{MR2077669}
Thomas Schick.
\newblock Real versus complex {$K$}-theory using {K}asparov's bivariant
  {$KK$}-theory.
\newblock {\em Algebr. Geom. Topol.}, 4:333--346, 2004.

\bibitem{MR2631479}
Antony~John Wassermann.
\newblock {\em Automorphic actions of compact groups on operator algebras}.
\newblock PhD thesis, University of Pennsylvania, 1981.

\bibitem{Witten:1998}
Edward Witten.
\newblock D-branes and {$K$}-theory.
\newblock {\em J. High Energy Phys.}, 1998(12):019, 1998, arXiv:hep-th/9810188.

\bibitem{Witten:1998-02}
Edward Witten.
\newblock Toroidal compactification without vector structure.
\newblock {\em J. High Energy Phys.}, (2):Paper 6, 43 pp. (electronic), 1998,
  arXiv:hep-th/9712028.

\end{thebibliography}
\end{document}